\documentclass[conference]{IEEEtran}
\usepackage{tabularx}
\usepackage[linesnumbered, ruled, vlined]{algorithm2e}
\usepackage{cite}
\usepackage{amsmath,amssymb,amsfonts}
\usepackage{algorithmic}
\usepackage{graphicx}
\usepackage{enumitem}
\usepackage{textcomp}
\usepackage{xcolor}
\usepackage{bm}
\usepackage{amsthm}
\usepackage{booktabs}
\usepackage{multirow}
\usepackage{threeparttable}
\usepackage{subfigure}
\usepackage{makecell}
\usepackage{caption} 
\theoremstyle{plain}

\newtheorem{theorem}{Theorem}
\newtheorem{corollary}{Corollary}
\theoremstyle{definition}
\newtheorem{definition}{Definition}
\theoremstyle{remark}
\newtheorem{remark}{Remark}
\usepackage{geometry}

\def\BibTeX{{\rm B\kern-.05em{\sc i\kern-.025em b}\kern-.08em
    T\kern-.1667em\lower.7ex\hbox{E}\kern-.125emX}}
\begin{document}
\newgeometry{left=0.673in, right=0.673in, top=0.760in, bottom=4.05cm}

\title{Variable-Length Joint Source-Channel Coding for Semantic Communication}

\author{\IEEEauthorblockA{Yujie Zhou\IEEEauthorrefmark{1}, Rulong Wang\IEEEauthorrefmark{1},  Yong~Xiao\IEEEauthorrefmark{1}\IEEEauthorrefmark{2}\IEEEauthorrefmark{3}, Yingyu Li\IEEEauthorrefmark{4}, Guangming~Shi\IEEEauthorrefmark{2}\IEEEauthorrefmark{3}\IEEEauthorrefmark{5}
\IEEEauthorblockA{\IEEEauthorrefmark{1}School of Elect. Inform. \& Commun., Huazhong Univ. of Science \& Technology, China}
\IEEEauthorblockA{\IEEEauthorrefmark{2}Peng Cheng Laboratory, Shenzhen, China}
\IEEEauthorblockA{\IEEEauthorrefmark{3}Pazhou Laboratory (Huangpu), Guangzhou, China}
\IEEEauthorblockA{\IEEEauthorrefmark{4}School of Mech. Eng. and Elec. Info., China University of Geosciences (Wuhan), China}
\IEEEauthorblockA{\IEEEauthorrefmark{5}School of Artificial Intelligence, Xidian University, China}}\vspace{-0.3in}}

\maketitle

\begin{abstract}
    This paper investigates a key challenge faced by joint source-channel coding (JSCC) in digital semantic communication (SemCom): the incompatibility between existing JSCC schemes that yield continuous encoded representations and digital systems that employ discrete variable-length codewords.
    It further results in feasibility issues in achieving physical bit-level rate control via such JSCC approaches for efficient semantic transmission.
    In this paper, we propose a novel end-to-end coding (E2EC) framework to tackle it.
    The semantic coding problem is formed by extending the information bottleneck (IB) theory over noisy channels, which is a tradeoff between bit-level communication rate and semantic distortion.
    With a structural decomposition of encoding to handle code length and content respectively, we can construct an end-to-end trainable encoder that supports the direct compression of a data source into a finite codebook.
	To optimize our E2EC across non-differentiable operations, e.g., sampling, we use the powerful policy gradient to support gradient-based updates.
    Experimental results illustrate that E2EC achieves high inference quality with low bit rates, outperforming representative baselines compatible with digital SemCom systems.
\end{abstract}
\begin{IEEEkeywords}
Semantic communication, joint source-channel coding, variable-length code, information bottleneck.
\end{IEEEkeywords}

\section{Introduction}
Semantic communication (SemCom) has revolutionized the communication paradigm, prioritizing the transmission of task-related information in messages rather than raw bit streams \cite{shi2021semantic}.
Previous studies have shown that, compared to the traditional paradigm, SemCom can achieve a significant improvement in communication efficiency for highly personalized task-specific services \cite{xiao2022imitation}, widely acknowledged as a critical enabler for the next-generation network.
In SemCom, signals only need to be recovered at the destination under semantic-level fidelity.
The typical separate source-channel coding is inefficient for SemCom since it primarily serves precise bit-level transmission.

To implement SemCom, joint source-channel coding (JSCC) is the critical enabling technique, which tackles data compression and channel coding together to achieve efficient semantic-oriented transmission \cite{10747747}.
Most recent works \cite{10747747, 8723589, xu2023deep} developed JSCC to achieve SemCom via neural networks with high non-linear expressivity.
This scheme is widely referred to as deep-JSCC, whose effectiveness has been extensively demonstrated through experiments in these works \cite{10747747, 8723589, xu2023deep} in various cases.
Despite its promising potential, most deep-JSCC methods still face a key challenge to achieve high communication efficiency in practical digital communication systems:
They are devoted to minimizing semantic distortion (loss function) through parametric continuous mappings based on neural networks (NNs).
In this case, the range of the corresponding encoding function is on a real-linear vector space rather than the canonical binary codebook in digital communication.
This mismatch makes the physical bit-level rate control based on them impossible, which is, however, a key factor affecting communication efficiency.

Several recent works \cite{shao2021learning, xie2023robust, rulongICC, 10845799} partially addressed this issue.
In works \cite{shao2021learning, xie2023robust, rulongICC}, the authors studied the rate control issue in deep-JSCC for digital SemCom from an information-theoretic perspective.
They applied the information bottleneck (IB) theory to yield the rate-distortion tradeoff in a SemCom paradigm.
However, in the experimental phases, their solutions still relied on empirical pruning or utilized a fixed-length codeword based on vector-quantization or Gumbel-Softmax, resulting in a rate far higher than the information-theoretic bound they claimed.
In work \cite{10845799}, the authors studied a special case of this problem.
They studied the combination of a deep source coding with a traditional digital channel coding by involving channel coding rate as an additional control variable.
In this case, the rate can be well-controlled by the two parts separately through a joint end-to-end semantic distortion.
While powerful, this method's applicability is limited since its end-to-end distortion requires a series of assumptions and heavy approximations, introducing bias.
Besides, the separate coding design of source and channel is not end-to-end (unlike its end-to-end distortion), leading to sub-optimality \cite{10747747}.

To tackle the key challenge, we establish a JSCC framework, namely, E2EC, to achieve the optimal tradeoff between the rate (w.r.t. source and channel) and the end-to-end semantic distortion in digital SemCom systems.
This tradeoff is formulated rigorously, extending the IB theory to involve channel impacts explicitly.
In which, we control the true code length rather than the information rate characterized by mutual information (MI) in IB to realize the bit-level efficient coding.
E2EC follows an end-to-end design, which contains an encoder-decoder pair that is fully implemented by NNs.
The output of the encoder and the input of the decoder are both semantic-level bit streams, which are compatible with digital architectures.
E2EC supports variable-length code by using the structural decomposition of encoding, i.e., a separate design of coding content and length. 
This design could enhance the feasibility of coding to achieve the information-theoretic bound (MI).
E2EC involves general sources and channels in probability. 
No specific data or channel model needs to be involved.
Key contributions are as follows:
\begin{itemize}
	\item We propose the E2EC framework, which forms a semantic rate-distortion over a noisy channel by extending the typical IB theory.
	As a deep-JSCC scheme, E2EC allows discrete variable-length code, highly compatible with real digital communication, bridging the gap between existing deep-JSCC approaches and practical systems.
	
	\item We develop the detailed structure of E2EC, present the intuition followed in its design, and forge potential connections with classical mathematical concepts.

    \item We show the optimization procedure of E2EC, which involves policy gradient, a.k.a. score function estimation, to support gradient-based training across non-differentiable operations, e.g., sampling from stochastic encoding and transmitting across unknown noisy channels.

	\item Experimental results verify the effectiveness of E2EC and also demonstrate the training dynamics and the statistical properties of the E2EC code.
\end{itemize}

\section{System Model and Problem Formulation}
\subsection{Notations}
We clarify the main notations used in this paper.
Uppercase letters (e.g., $X$) denote random variables or constants. Lowercase letters (e.g., $x$) denote elements of sets. The distribution induced by a random variable $X$ is represented by $P_X$, and its probability density function (pdf) or probability mass function (pmf) by $p_X$. The Kullback-Leibler (KL) divergence between $P$ and $Q$ is given by $D_{\text{\rm KL}}(P\|Q)\triangleq  \mathbb{E}_{P}[\log \frac{dP}{dQ}]$ with $P\ll Q$. 
$[R]$ abbreviates a set $\{1,2,...,R\}$ for any positive integer $R$.
For notational comfort, we indiscriminately adopt $dx \triangleq \delta(dx)$ or $\mathrm{Leb}(dx)$ to allow the uniform use of ``$\int dx$'' for integration and summation in suitable cases without ambiguity.

\subsection{System Model}
We consider a SemCom system involving a transmitter
and a receiver w.r.t. a data source, as illustrated in Fig. \ref{fig:1}. 
The data source is denoted by $S\triangleq (X, Y)$, where $Y$ stands for intrinsic semantics while $X$ stands for raw data conditioned on $Y$, i.e., $X \sim P_{X|Y}$.
Notice that $Y$ is unobservable, and only $X$ can be accessed by the transmitter \cite{liu2022indirect}.
The encoder, deployed on the transmitter, is a mapping $f: x \in X \mapsto z \in Z$, transforming the high-dimensional $X$ to a variable-length compressed $Z$ as the channel input with the $f$-induced probability kernel $P_{Z|X}$.
In this paper, due to the digital SemCom paradigm, we assume that $Z$ is supported on a finite binary $\{0,1\}^{R_{\max}}$ with $R_{\max} < \infty$ rather than the default $\mathbb{R}$-linear vector space.
We consider a channel $W: z \in Z \mapsto \hat{z} \in \hat{Z}$ where $\hat{Z}$, the corrupted output of the channel, also lies in $\{0,1\}^{R_{\max}}$. 
The main objective of the decoder of the receiver is to obtain a recovered semantic signal $\hat Y$ given the channel output $\hat{Z}$, where the decoder is a mapping $g: \hat z \in {\hat Z} \mapsto \hat y \in \hat Y$.
The SemCom system maintains a Markov chain $Y \leftrightarrow X \leftrightarrow Z \leftrightarrow {\hat{Z}} \leftrightarrow {\hat{Y}}$.

\subsection{IB-induced Semantic Distortion over A Noisy Channel}
\begin{figure}[t]
	\centering
	\includegraphics[width=1\linewidth]{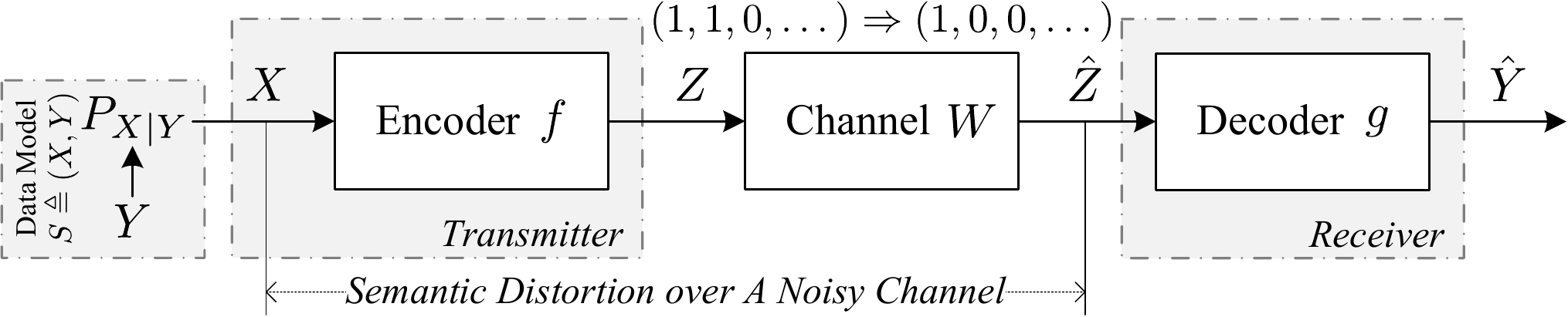}
	\caption{Illustration of a digital SemCom system. The channel is defined on bit sequences, absorbing digital modulation/demodulation mechanisms \cite{xie2023robust}.}
	\label{fig:1}
	\vspace{-0.2in}
\end{figure}
The distortion measure of SemCom is known as the logarithmic loss of remote source coding \cite{6651793}.
This stems from the IB theory.
Let $Y \leftrightarrow X \leftrightarrow Z$, we present the IB expression 
\begin{equation}\label{obj:ib}
	\max\nolimits_{P_{Z|X}: I(X;Z)\triangleq D_{\text{\rm KL}}(P_{{X}{Z}}\|P_{X}P_{Z})\leq \epsilon} I(Z;Y), 
\end{equation}
where $I(\cdot;\cdot)$ denotes mutual information (MI).
The expression (\ref{obj:ib}) means that we maximize $I(Z;Y)$, the informativeness of coded representation $Z$ corresponding to the intrinsic semantics $Y$, under a given $I(X;Z)$, i.e., the information rate that $Z$ can achieve \cite{cover1999elements}.
The measure $I(Z;Y)$ naturally captures the information loss about semantics, which is popular and widely applied in existing works.
To seamlessly extend it to the JSCC case, we rewrite $I(Z;Y)$ into a typical distortion form, i.e.,
\begin{equation}\label{eq:distib}
	I(Z;Y) \Rightarrow d_{\text{\rm IB}}(X, Z) \triangleq D_{\text{\rm KL}}(P_{Y|X}\|P_{Y|Z}), 
\end{equation}
where the average distortion recovers $I(Z;Y)$ with a constant gap, i.e., $\mathbb{E}_{{XZ}}[d_{\text{\rm IB}}(X, Z)] = -I(Z;Y) + I(X;Y)$ with $0 \leq I(X;Y) = \mathrm{const}$ due to the invariant data distribution $S$.

By introducing the channel impact $W$ on $Z$, we can define the IB-based distortion over a noisy channel, i.e.,
\begin{definition}
	The semantic distortion measure induced by the IB principle (\ref{eq:distib}) over a noisy channel is as  \vspace{-0.025in}
	\begin{equation}
		d_{W}(X, Z) \triangleq D_{\text{\rm KL}}(P_{Y|X}\|P_{Y|W(Z)}),\vspace{-0.025in}
	\end{equation}
	which is random w.r.t. $W$. 
	Because the full statistics of $d_{W}$ is complicated, the expected version is often considered\cite{9046817} \vspace{-0.025in}
	\begin{equation}
		d_{\bar W}(X, Z) \triangleq \mathbb{E}_{W}[d_{W}(X, Z)],\vspace{-0.025in}
	\end{equation}
	where $\mathbb{E}_{W}[d_{W}(X, Z)]$ is equal to $\mathbb{E}_{\hat{Z}|Z}[D_{\text{\rm KL}}(P_{Y|X}\|P_{Y|\hat{Z}})]$ by characterizing the random transformation $W$ as a kernel $P_{\hat{Z}|Z}$.
\end{definition}
Semantic distortion (SD) is end-to-end \cite{10845799} and well-defined inheriting the property from $d_{\text{\rm IB}}$ that successfully quantifies the information loss about semantics $Y$ in the channel-corrupted  $\hat{Z}=W(Z)$ relative to the original data $X$ under noisy channel conditions.

\subsection{Problem Formulation}
This paper aims to solve the IB-induced JSCC problem for SemCom, which is expressed as follows\vspace{-0.025in}
\begin{equation}\label{obj:main}\tag{\textbf{P}}
	\min_{f: R\leq \epsilon, \ \mathrm{supp}(Z) \subseteq \{0,1\}^{R_{\max}}} \mathbb{E}_{XZ}[d_{\bar W}(X, Z)], \vspace{-0.025in}
\end{equation}
where $R$ is the rate of $Z = f(X)$, equivalently, it is the average length of the code \cite{cover1999elements}, i.e., $R \triangleq \mathbb{E}_Z [l(Z)]$ with $l(Z) \leq R_{\max}$.
$l(Z)$ represents the length of $Z$.
$R \geq I(X;Z)$ is the practical communication rate, replacing the asymptotically ideal MI in (\ref{obj:ib}) \cite{cover1999elements}.
$\epsilon$ is often  set to be less than the channel capacity.
By tackling this problem, we find a good encoder $f$ which encodes $X$ into a discrete low-rate $Z$ while achieving low SD.
Note that the decoder $g$ is determined via the chain $Y \leftrightarrow X \leftrightarrow Z \leftrightarrow {\hat{Z}}$ if $f$, $S$ and $W$ are known, i.e., $p(y|\hat{z})\!=\!\int dxdz p(x, y, z|\hat{z})\!=\!\int \!dxdz p(\hat{z}|z){p(x, y)p(z|x)}/{p(\hat{z})}$ in which $p(z|x)$ is induced by $f$, $p(y|\hat{z})$ by $g$, $p(\hat{z}|z)$ by $W$, and $p(x,y)$ by $S$.

Two fundamental challenges make this problem (\ref{obj:main}) difficult to solve:
(i) The continuous-to-discrete mapping $f$ is hard to model, especially in the case that variable-length encoding is allowed.
This discreteness further leads to the feasibility issue of differentiable optimization.
(ii) Directly computing the ideal decoder $g$ w.r.t. $f$ through the Markov chain is intractable due to the unknown distributions of source and channel.

\section{The E2EC Framework}
In this section, we propose a computationally feasible solution for (\ref{obj:main}) through our end-to-end E2EC framework.

\vspace{-0.025in}
\subsection{Encoding with Structural Decomposition}
The key characteristic of digital communication lies in the discreteness and variable length of the coding.
We develop an encoding function as follows to match it.

The encoding $f$ is split into two functional modules, namely, $f_l$ and $f_{\tilde{z}}$. 
To be specific, the first module $f_l: x \in X \to l \in [R_{\max}]$ decides code length.
The second module $f_{\tilde{z}}: x \in X \to {\tilde{z}} \in \{0,1\}^{R_{\max}}$ decides code content on the augmented coding space $\{0,1\}^{R_{\max}}$.
By truncating such ${\tilde{z}}$ to its first $l$ bit, we obtain the variable-length code $z = {{\tilde{z}}}{[1:l]}$.
In this case, $Z$ can be supported on a subset of $\{0,1\}^{R_{\max}}$ by excluding all zero-measure singletons.
Therefore, by introducing the truncation operation $\mathrm{truncate}(\cdot,\cdot)$, which is surjective from $({\tilde{z}}, l)$ to $z$, we finally obtain $f(x) = \mathrm{truncate}((f_l,f_{\tilde{z}})(x)) = f_{\tilde{z}}(x)[1\!:\!f_l(x)]$.
In fact, intermediate variables ${\tilde{z}} \sim {\tilde Z}$ and $l \sim L$ can be inserted into the chain between $X$ and $Z$ as $X\leftrightarrow ({\tilde Z}, L)\leftrightarrow Z$.
\begin{remark}\label{remark:1}
	The intuition behind this construction: heuristically, we expect that there exists a $Y$-related semantic decomposition over raw data $X$, analogous to the orthogonal decomposition or the Fourier transform, satisfying that each component provides \emph{disentangled and progressively refined semantic information} in an ascending order.
	Correspondingly, ${\tilde{Z}}$ represents all non-zero coefficients and $\mathrm{truncate}(\cdot,\cdot)$ equipped with stopping criterion $L$ stands for the ``lowpass filtering'':
	if the first $L$ components, i.e., coefficients $(\tilde{Z}^i)$ with their respective implicit ``orthogonal bases'', are sufficient to recover $Y$, then we extract the principal information via the filter as the code. 
	This hypothesis is further empirically explored at Section \ref{sec:exp}.

	From another viewpoint, this structural decomposition also rigorously follows the construction of atomic random measures \cite[Example VI.1.7]{ccinlar2011probability}, which is typical for characterizing point processes, if the collection of all bits of $Z$ is an independency given $X$.
    Therefore, it is a type of structured random prior.
\end{remark}

\vspace{-0.05in}
\subsection{Decoding with One-to-One Embedding}\vspace{-0.025in}
$p(y|\hat{z})$ is computationally intractable.
We use a learnable $g$ corresponding to a randomly initialized probability $q(y|\hat{z})$ to replace $p(y|\hat{z})$.
This corresponds to the traditional variational approximation technique \cite{shao2021learning} that recovers $p(y|\hat{z})$ with $q(y|\hat{z})$.

This $g$ should be carefully designed due to the discreteness of the corrupted $\hat{Z}$.
As previously mentioned in Remark \ref{remark:1}, we priorly assume that $Z$ corresponds to a series of low-frequency coefficients of the implicit decomposition.
We further assume that $\hat{Z}= W(Z)$ holds the same property.
If we know the bases of $\hat{Z}$, we can successfully recover $X$ in a semantic sense under $d_{\bar W}$.
Then, we can infer $\hat{Y}$ of $Y$ through such a recovered $\hat{X}$.

The true bases of $\hat{Z}$ are unknown.
We directly estimate the corresponding implicit decomposition by components: 
for any $\hat{z} \sim \hat{Z}$, the $i$-th basis vector with its coefficient is approximated with a vector $e^i \triangleq e^i(\hat{z}^i) \in \mathbb{R}^d$ for the $i$-th bit of $\hat{z}$\footnote{In fact, this semantic decomposition is not standard, which allows an affine dependence, i.e., $e^i(\hat{z}^i) = \hat{z}^i e^i_1 + e^i_0$, $\hat{z}^i \in \{0,1\}$, $e^i_1$ for $1$ and $e^i_0$ for $0$.}.
Then, we write $\hat{x}= \sum_i e^i(\hat{z}^i)$ to be the semantically recovered data: it should ensure that $p(y|\hat{x})$ is close to $p(y|x)$ in KL divergence derived from SD.
It is a one-to-one embedding of the corrupted code $\hat{z}$ since for each component $\hat{z}^i \in \{0,1\}$, we map it one-to-one into an $\mathbb{R}^d$-representation $e^i$.
We aim to generate non-collapsed semantic representations in so doing, similar to \cite{lin2021learning}.

To summarize, this decoding $g$ is split into two functional modules, i.e., $g_{\hat{x}}$ and $g_y$. 
The first module $g_{\hat{x}}: \hat{z} \in \hat{Z} \to \hat{x} \in \hat{X}$ semantically reproduces $X$ with $\hat{x}= g_{\hat{x}}(\hat{z}) = \sum_i e^i(\hat{z}^i)$.
The second module $g_y: \hat{x} \in \hat{X} \to \hat{y} \in \hat{Y}$ infers the reproduction of $Y$ with $q(y|\hat{x})$.

\vspace{-0.025in}
\subsection{Transformed Objective}\vspace{-0.025in}
In the preceding subsection, we design the variational $q(y|\hat{z})$ induced by $g = g_y \circ g_{\hat{x}}$ in detail. 
Nevertheless, replacing $p(y|\hat{z})$ with $q(y|\hat{z})$ can implicitly affect the calculation of SD.
In fact, the replacement yields a variational upper-bound for SD.
\vspace{-0.05in}
\begin{theorem}\label{thm:vibo}
	The following variational bound holds for SD:\vspace{-0.025in}
	\begin{equation}
		\mathbb{E}_{XZ}[d_{\bar W}(X, Z)] \leq \mathbb{E}_{\hat{Z}}[H(P_{Y|\hat{Z}}, Q_{Y|\hat{Z}})],\vspace{-0.025in}
	\end{equation}
    where $H(P_{Y|\hat{Z}}, Q_{{{Y}}|\hat{Z}}) \triangleq -\int p(y|\hat{z})\log {q(y|\hat{z})}dy$ is the conditional cross entropy. $p(y|\hat{z})$ is w.r.t. $P_{Y|\hat{Z}}$ and $q(y|\hat{z})$ is w.r.t. $Q_{Y|\hat{Z}}$. The equality holds iff $p(y|\hat{z}) = q(y|\hat{z})$.
\end{theorem}\vspace{-0.1in}
\begin{proof}
	The proof is postponed to Appendix.
\end{proof}

By introducing a multiplier $\lambda$ of the Lagrangian of (\ref{obj:main}) with this bound, we obtain a computable objective:\vspace{-0.05in}
\begin{equation}\label{obj:Lagrangian}
	\mathcal{L}[f,g]=\mathbb{E}_{\hat{Z}}[H(P_{Y|\hat{Z}}, Q_{Y|\hat{Z}})] + \lambda R,\vspace{-0.05in}
\end{equation}
which is the main objective to minimize in this paper.
In the remainder of this section, we instantiate $f$ and $g$ as optimizable parameterized functions and present the optimization process.

\vspace{-0.1in}
\subsection{Probability Parameterization}\vspace{-0.05in}
\begin{algorithm}[t]
	\caption{Training Procedure of E2EC}
	\label{alg:E2EC}
	\KwIn{source $S$, channel $W$, multiplier $\lambda$.}
	\KwOut{optimized parameters $\theta$, $\xi$, $e$, and $\phi$.}
	\While{\textnormal{not converged}}{
		\tcp{Sample Collection}
		randomly sample $\left(x_j,y_i\right)$ from $S$;

		encode $z_j = f(x_j)$ w.r.t. $\theta$ and $\xi$, specifically, randomly sample $l_j \sim p(l|x_j; \theta)$ and $\tilde{z}_j \sim p(\tilde{z}|x_j;\xi)$, then $z_j = \mathrm{truncate}(l_j, \tilde{z}_j)$;

		generate channel-corrupted $\hat{z}_j = W(z_j)$;

		infer $\hat{y}_j = g(\hat{z}_j)$ w.r.t. $e$ and $\phi$, specifically, randomly sample $\hat{y}_j \sim q(y|\hat{z}; e, \phi)$;

		collect $N$ samples as $(x_j,y_j, l_j, \tilde{z}_j, \hat{z}_j)$ to approach $p(x,y,l,\tilde{z},\hat{z})$ by $\hat{p}(x,y,l,\tilde{z},\hat{z})$;
		
		\tcp{Gradient Estimation}
		estimate $\nabla_{\theta,\xi} D$ via policy gradient (\ref{eq:policy}), while estimate $\nabla_{e,\phi} D$ (\ref{eq:grad}) and $\nabla_{\theta} R$ (\ref{eq:rate}) via backprop with $\hat{p}(x,y,l,\tilde{z},\hat{z})$;

		update $\theta$, $\xi$, $e$, and $\phi$ via gradient descent towards minimizing $\mathcal{L}=D + \lambda {R}$.
	}
\end{algorithm}
\setlength{\textfloatsep}{0pt}
We parameterize the learnable functions $f$ and $g$ in terms of neural networks (NNs) that have high expressivity \cite{yarotsky2017error}.

Specifically, we define $N\!N$ as a neural network.
We clarify a common abbreviation throughout this paper, i.e., for any given $y$, $p(y|x) \equiv p_{Y|X}(y|x)$.
For encoding, we construct $N\!N_\theta: x \in X \mapsto p_{L|X}(\cdot|x) \in \mathcal{P}_{[R_{\max}]}$, which forms a conditional probability, i.e., the output of this $\theta$-NN corresponds a distribution of the length variable $L$ conditioned on a given $x$.
Subsequently, we instantiate $N\!N_\xi: x \in X \mapsto p_{\tilde{Z}|X}(\cdot|x) = \prod_i p_{\tilde{Z}^i|X}(\cdot|x) \in \mathcal{P}_{\{0,1\}^{R_{\max}}}$, in which the component $(\tilde{Z}^i)$ is an independency conditioned on $X$, following the construction in Remark \ref{remark:1}.
We develop the encoding $f$ as
$z = f(x) = \mathrm{truncate}(f_l(x), f_{\tilde{z}}(x))$ with $l= f_l(x) \sim p(l|x;\theta)$, i.e., $l$ is sampled from $N\!N_\theta(x)$, and ${\tilde{z}} = f_{\tilde{z}}(x) \sim p(\tilde{z}|x;\xi)$, i.e., ${\tilde{z}}$ is sampled from $N\!N_\xi(x)$.

For decoding, we develop the embedding module as previously introduced, i.e., $g_{\hat{x}}$ with the learnable $e = (e^i)_{i\in [R_{\max}]}$.
We also construct $N\!N_\phi: g_{\hat{x}}(\hat{z}) = \hat{x} \in \hat{X} \mapsto q_{Y|\hat{X}}(\cdot|\hat{x}) \in \mathcal{P}_{Y}$.
We have $g$ as $\hat{y}= g(\hat{z}) = g_y \circ g_{\hat{x}}(\hat{z}) \sim q(y|\hat{z}; e, \phi)$, i.e., $y$ is sampled from $N\!N_\phi(g_{\hat{x}}(\hat{z}))$ with $\hat{z} = W(z)$.\vspace{-0.05in}
\begin{corollary}\label{cor:parambo} 
	With probability parameterization which implies the conditional independence $L\perp \tilde{Z}|X$, the following objective is equivalent to the Lagrangian (\ref{obj:Lagrangian}), i.e.,\vspace{-0.05in}
	\begin{equation}
		\begin{aligned}
			\mathcal{L} = & \Big\{D \triangleq - \mathbb{E}_{X} \mathbb{E}_{L|X;\theta}\mathbb{E}_{\tilde{Z}|X;\xi}\mathbb{E}_{Y|X} \mathbb{E}_{\hat{Z}|L\tilde{Z}}[\log {q(Y|\hat{Z};e, \phi)}]\Big\}
			\\ & + \lambda \Big\{{R} = \mathbb{E}_Z [l(Z)] = \mathbb{E}_\theta[L] = \mathbb{E}_{X} \mathbb{E}_{L|X;\theta}[L]\Big\}.
		\end{aligned}
	\end{equation}
\end{corollary}
\begin{proof}
	The proof is postponed to Appendix.
\end{proof}

\subsection{Optimization}
The training procedure is summarized in Algorithm \ref{alg:E2EC}, where the unknown $S$ and $W$ are only accessible through sampling.
The critical technical trick for gradient estimation is given in the following.
The calculation of $\mathcal{L}$ involves non-differentiable operations, including the sampling from parameterized $\theta$- and $\xi$-NNs and esp. from the unknown channel $W$ and the truncation $\mathrm{truncate}(\cdot,\cdot)$.
Therefore, this results in the untraceability of gradients, making it impossible to achieve efficient gradient backpropagation in training directly.


Specifically, $\nabla_{e, \phi} D$ and $\nabla_\theta {R}$ can be successfully calculated through backpropagation, but $\nabla_{\theta, \xi} D$ cannot.
In the following, we elaborate on how to compute such intractable terms.
With $\nabla_{\theta, \xi} D = (\nabla_{\theta} D, \nabla_{\xi} D)$, explicitly calculate $\nabla_\theta D$ as follows:
\begin{equation}\label{eq:policy}
    \!\!\!
    \begin{aligned}
        \nabla_\theta D & = - \mathbb{E}_{X} \nabla_\theta \mathbb{E}_{L|X;\theta}\underbrace{\mathbb{E}_{\tilde{Z}|X;\xi}\mathbb{E}_{Y|X} \mathbb{E}_{\hat{Z}|L\tilde{Z}}[\log {q(Y|\hat{Z};e, \phi)}]}_{\triangleq \mu(L,X)}
        \\ & = - \int p(x) \nabla_\theta p(l|x; \theta)\mu(l,x)dxdl
        \\ & = - \int p(x) p(l|x; \theta)\nabla_\theta \log p(l|x; \theta)\mu(l,x)dxdl
        \\ & \approx - \frac{1}{N}\sum_{j=1}^N \mu_j(l_j,x_j)\nabla_\theta \log p(l_j|x_j; \theta)
        \\ & = \nabla_\theta \Bigg[- \frac{1}{N}\sum_{j=1}^N \log q(y_j|\hat{z}_j)\log p(l_j|x_j; \theta)\Bigg].
    \end{aligned}
    \!\!\!\!\!\!\!\!\!
\end{equation}
The approximation is conducted via Gibbs sampling \cite{shao2021learning} over the joint distribution of all variables through the Markov chain, allowing the replacement of $p(x,y,l,\tilde{z},\hat{z})$ with the empirical $\hat{p}(x,y,l,\tilde{z},\hat{z}) = \frac{1}{N}\sum_{j}^{N}\delta_{x_j}(x)\delta_{y_j}(y)\delta_{l_j}(l)\delta_{\tilde{z}_j}(\tilde{z})\delta_{\hat{z}_j}(\hat{z})$.
So $\mu_j$ abbreviates the sampling of $y,\tilde{z},\hat{z}$; $\hat{p} \rightharpoonup p$ a.s. as $n \to \infty$.

This effective score function estimation for gradients is also called \emph{policy gradient} following canonical policy optimization \cite{sutton1999policy}.
It enables us to estimate the gradient from samples across non-differentiable operations.
For $\nabla_\xi D$, we do so also.

On the other hand, for $\nabla_{e, \phi} D$, we typically estimate
\begin{equation}\label{eq:grad}
    \nabla_{e, \phi} D \approx  \nabla_{e, \phi} \Bigg[- \frac{1}{N}\sum_{j=1}^N\log q(y_j|\hat{z}_j;e, \phi)\Bigg].
\end{equation}
For $\nabla_{\theta} R$, we also have
\begin{equation}\label{eq:rate}
    \nabla_{\theta} R \approx  \nabla_{\theta} \Bigg[- \frac{1}{N}\sum_{j=1}^N \int l\cdot p(l|x_j;\theta)dl\Bigg].
\end{equation}

According to (\ref{eq:policy})-(\ref{eq:rate}), we successfully estimate the gradient of the Lagrangian $\mathcal{L}$ w.r.t. the parameters $\theta$, $\xi$, $e$, and $\phi$.
Then, we use the typical gradient descent to find the stationary point.

\vspace{-0.025in}
\subsection{Further Discussion}\vspace{-0.025in}
The E2EC framework is now well-established. Let us clarify its computational feasibility, as claimed at the start of this section. 
Besides variational approximation, this crucial property is also enabled by the structural decomposition in the encoding.
Without decomposition, $f$ must directly map $X$ to $Z$, in which $Z$ has $2^{R_{\max}+1}-2$ values, as each one corresponds to a valid code of length $l \leq R_{\max}$. In this case, a standard categorical implementation of $f$ holds an output dimension exponential in $R_{\max}$, leading to high space complexity that harms trainability.
As Remark \ref{remark:1} states, by leveraging the prior structure of length-content separation, we reduce the dimensional dependence on $R_{\max}$ to linear.

\section{Performance Evaluation}\label{sec:exp}\vspace{-0.025in}
\subsection{Experimental Setup}\vspace{-0.025in}
We evaluate E2EC on the MNIST dataset, which contains 70,000 handwritten digit images in 10 classes, each consisting of 28$\times$28 gray-scale pixels. 
We utilize the digit class as $Y$ and the corresponding pixel values as $X$.

We consider a binary symmetric channel (BSC) with an error probability $p_e$ to model the digital communication channel. $p_e$ acts on the bit stream, i.e., each component of a codeword.

The encoder is composed of $f_l$ and $f_{\tilde{z}}$, both implemented as feedforward ReLU networks with layer normalization.
Especially, $f_l$ maps $X$ to a sample of a categorical distribution, $f_{\tilde{z}}$ maps $X$ to a sample of a product of Bernoulli distributions.

The decoder is composed of $g_{\hat{x}}$ and $g_y$. $g_{\hat{x}}$ is implemented as a combination of a query on a collection of trainable $\mathbb{R}$-valued vectors for channel-corrupted bit streams and a sum operation.
One-to-one embeddings are initialized through a series of $d$-dimensional i.i.d. variables with an identical isotropic normal distribution to ensure the (initial) orthogonality in an asymptotical sense.
$g_y$ is a classifier implemented as a feedforward ReLU network with layer normalization.

We compare E2EC with the following baselines:
\begin{itemize}
    \item \textbf{digital deep-JSCC:} This represents a series of traditional deep-JSCC schemes for digital SemCom, which uses the vector quantization and Gumbel-softmax to learn discrete representations with fixed-length codewords, e.g., \cite{xie2023robust, rulongICC}. 
    \item \textbf{RDBO:} The rate-distortion bound (RDBO) that semantic transmission can achieve.
	It is implemented by the typical deep variational IB.
	RDBO does not involve any channel condition, which corresponds to the ideal result.
\end{itemize}
 
We use a series of technical tricks to facilitate the convergence of training. 
{For details, please refer to our online code}\footnote{Due to space limitation, detailed configurations and engineering tricks are omitted, e.g., variance reduction and symbol grounding \cite{lin2021learning}. 
Our experimental code is available at
https://github.com/SamuChamp/E2ECframework.}.

\subsection{Numerical Results}
\begin{figure}[t]
		\includegraphics[width=\linewidth]{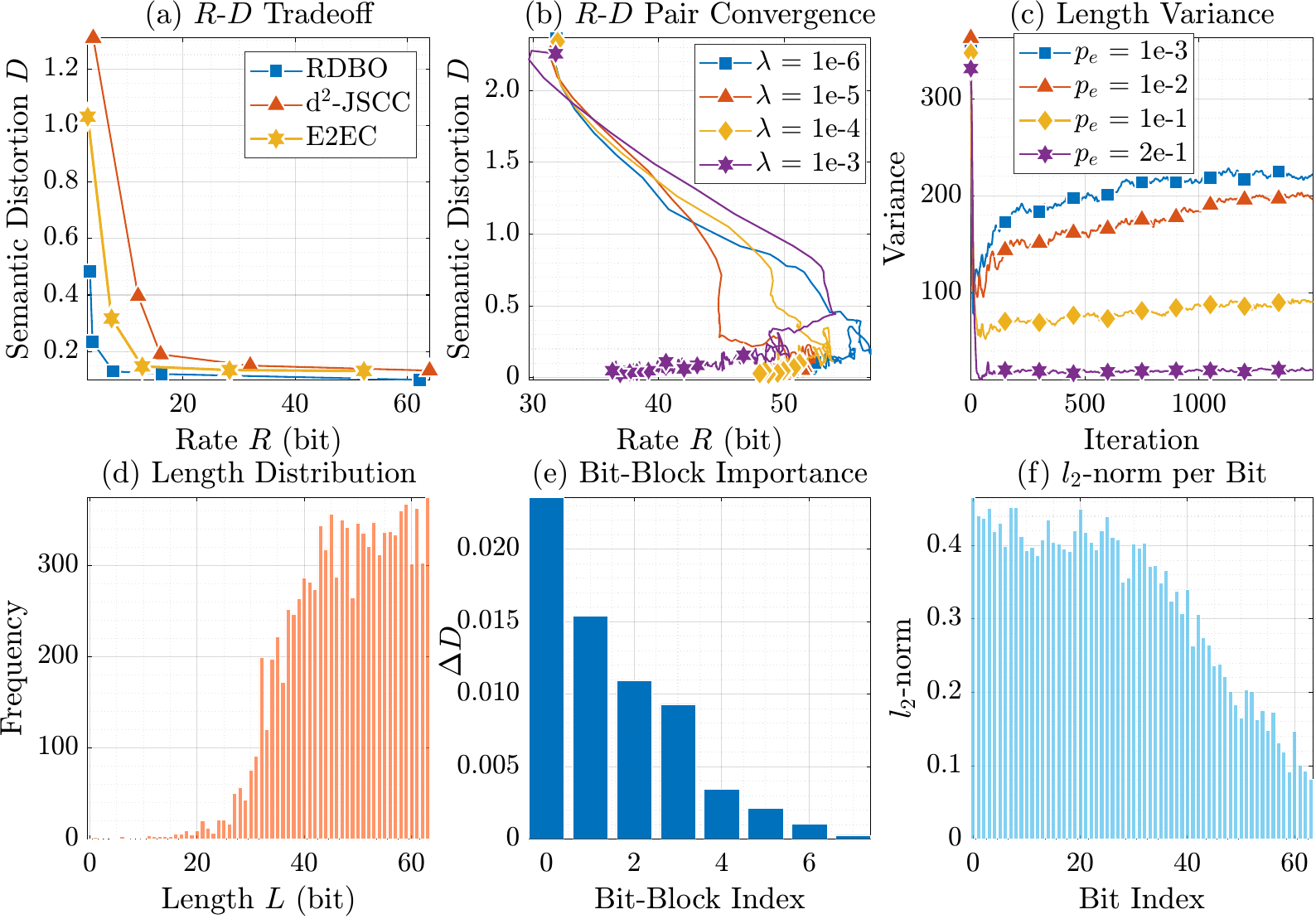}
		\caption{
			Numerical results: (a) The rate-semantic distortion tradeoff over a BSC. (b) The convergence of E2EC's $R$-$D$ curve on the information plane. 
			(c) The evolution of $\mathrm{Var}(L)$, during training. 
			(d) The non-normalized empirical $\hat{P}_L$, where frequency means the number of samples with the  normalization constant 10,000 (the size of testing dataset). 
			(e) The increment of $D$, i.e., $\Delta D$ varying by discarding equally divided bit-blocks. (f) The $l_2$-norm of 0-1 embeddings per bit, i.e., $\|e_1^i\|_2 = \|(1\times e_1^i+e_0^i)-(0\times e_1^i+e_0^i)\|_2, \forall i \in [R_{\max}]$.
		}
    \label{fig:results}
\end{figure}
As depicted in Fig. \ref{fig:results}, we evaluate the performance of E2EC, the training dynamics, and the properties of the length variable in (a)-(d).
Also, we empirically illustrate the intuition provided in Remark \ref{remark:1} in (e) and (f).
The default settings are as follows: (i) E2EC: $R_{\max}$=64, $d$=64, and $\lambda$=1e-6; (ii) BSC: $p_e$=1e-1.

Specifically in Fig. \ref{fig:results}(a), due to the introduction of variable-length coding, it is obvious that E2EC's coding efficiency far exceeds that of the traditional Gumbel-softmax-based d$^2$-JSCC (digital deep-JSCC) with high flexibility.
\emph{In contrast to the d$^2$-JSCC, which requires explicitly specifying code length, E2EC can flexibly balance rate and distortion through the additional degree of freedom $\lambda$ given the maximal length.}
Besides, as the rate $R$ increases, both E2EC and d$^2$-JSCC gradually approach the optimal RDBO at the default BSC.

In Fig. \ref{fig:results}(b), we study the convergence dynamics of the $R$-$D$ pairs on the information plane varying with $\lambda$.
The trajectories start from the top left corner, move rapidly to the bottom right, and finally converge at the bottom left, where there are denser markers. This is in line with IB's information plane dynamics.

In Fig. \ref{fig:results}(c) and (d), the variable-length characteristics of the learned codedwords are depicted. 
As $p_e$ decreases, the length variation increases to better match the MI boundary.

Fig. \ref{fig:results}(e) and (f) empirically verify the disentanglement and progressive refinement of semantic information.
In particular, (e) reflects that the lower bit of the code is more semantically informative. (f) verifies this view by the Euclidean distance of 0-1 embeddings: the closer they are, the more difficult it is to distinguish between 0 and 1, and hence the less the effective information content.
Disentanglement, i.e., orthogonality, holds since \emph{the maximum absolute inner product of cross embeddings} is very small, i.e., {0.0480} at the stationary point ({0.1718} at the initial point).
We also present the performance of E2EC as $p_e$ and $R_{\max}$ vary in Table \ref{table:perform}. This shows the rate-adaptability of E2EC and its sensitivity to the critical parameter $R_{\max}$.
\begin{table}[t]
    \centering
    \begin{threeparttable}
        \setlength{\tabcolsep}{0.09in}
        \caption{
			Rate-accuracy pair varying with BSC error probability and maximal code length under the default setting (95\% confidence interval).
		}\label{table:perform}
        \begin{tabular}{c|cccc}
            \toprule
            \multirow{1}{*}{$p_e$}
            & 2e-1 & 1e-1  & 1e-2  & 1e-3          \cr
          
            \midrule\multirow{1}{*}{$R$ (bit)}
            & 57.79$\pm$0.09 & 48.22$\pm$0.15 & 39.37$\pm$0.12  & 37.84$\pm$0.13       \cr

			\midrule\multirow{1}{*}{ACC (\%)}
            & 98.02$\pm$0.27 & 97.97$\pm$0.28 & 98.18$\pm$0.26  & 98.41$\pm$0.25       \cr
            \bottomrule
        \end{tabular}
    \end{threeparttable}\vspace{0.05in}
    \begin{threeparttable}
        \setlength{\tabcolsep}{0.09in}
        \begin{tabular}{c|cccc}
            \toprule
            \multirow{1}{*}{$R_{\max}$}
            & 8 & 16  & 32  & 64         \cr
          
            \midrule\multirow{1}{*}{$R$ (bit)}
            & 7.20$\pm$0.03 & 15.65$\pm$0.01 & 27.23$\pm$0.06  & 48.22$\pm$0.15       \cr

			\midrule\multirow{1}{*}{ACC (\%)}
            & 92.64$\pm$0.51 & 95.56$\pm$0.40 & 97.72$\pm$0.29  & 97.97$\pm$0.28       \cr
            \bottomrule
        \end{tabular}
    \end{threeparttable}
\end{table}

\vspace{-0.005in}
\section{Conclusion}\vspace{-0.005in}
This paper proposes E2EC, an end-to-end coding framework for digital SemCom that addresses the mismatch between existing deep-JSCC methods and practical digital communication architectures.
The key innovation of E2EC lies in its capability to yield discrete variable-length binary codes directly through a separate design of code length and content.
This design enables E2EC to achieve the real bit-level rate control while preserving NN expressivity in semantic understanding to realize efficient transmission in digital SemCom systems.
Theoretical insights are provided to offer intuitive explanations of the E2EC workflow's logic.
Optimization across non-differentiable operations is established to ensure the E2EC's trainability. 
Experimental results show the significant advantage of the end-to-end coding design for digital SemCom versus representative baselines.
We hope that this work will provide new ways and insights for the practical design of efficient digital SemCom systems.


\vspace{-0.005in}
\appendix\vspace{-0.005in}
\begin{proof}[Proof of Theorem \ref{thm:vibo}]
	Directly bound the SD term to obtain\vspace{-0.05in}
	\begin{equation}
		\begin{aligned}
			& \mathbb{E}_{XZ}[d_{\bar W}(X, Z)] = \mathbb{E}_{XZ}[\mathbb{E}_{\hat{Z}|Z}[D_{\text{\rm KL}}(P_{Y|X}\|P_{Y|\hat{Z}})]]
			\\ & = \int dxdydzd\hat{z} p(x, z)p(\hat{z}|z) p(y|x)\log \frac{p(y|x)}{p(y|\hat{z})}
			\\ & = \bigg\{\int dxdy p(x)p(y|x)\log {p(y|x)} \leq 0 \bigg\}
			\\ & \quad - \int dxdydzd\hat{z} p(z|x)p(\hat{z}|z) p(y,x)\log {p(y|\hat{z})}
			\\ & \leq - \int dyd\hat{z} p(\hat{z})p(y|\hat{z})\log {p(y|\hat{z})}
			\\ & \leq - \int dyd\hat{z} p(\hat{z})p(y|\hat{z})\log {p(y|\hat{z})}
			\\ & \quad + \int dyd\hat{z} p(\hat{z})p(y|\hat{z})\log \frac{p(y|\hat{z})}{q(y|\hat{z})} \; (D_{\text{\rm KL}}(\cdot\|\cdot) \geq 0)
			\\ & = - \int dyd\hat{z} p(\hat{z})p(y|\hat{z})\log {q(y|\hat{z})}
		\end{aligned}
		\vspace{-0.05in}
	\end{equation}
	We conclude this proof with the pre-defined notation of the conditional cross entropy.
\end{proof}
\begin{proof}[Proof of Corollary \ref{cor:parambo}]
	Recall $L\perp \tilde{Z}|X$ and $X\leftrightarrow ({\tilde Z}, L)\leftrightarrow Z$, then we have
	\begin{equation}
		\begin{aligned}
			& \mathbb{E}_{\hat{Z}}[H(P_{Y|\hat{Z}}, Q_{Y|\hat{Z}})]
			\\ & = \!- \! \int \!\! dxdydzd\hat{z}dld\tilde{z} p(z|l,\tilde{z})p(l,\tilde{z}|x)p(\hat{z}|z) p(y,x)\log {q(y|\hat{z})}
			\\ & = \!- \! \int \!\! dxdyd\hat{z} dld\tilde{z}p(x)p(y|x)p(l|x)p(\tilde{z}|x) p(\hat{z}|l,\tilde{z}) \log {q(y|\hat{z})}
		\end{aligned}
	\end{equation}
	The last equality holds due to the degeneration of the transition $p(z|l,\tilde{z}) = \delta_{\tilde{z}[1:l]}(z)$ corresponding to the truncation operation.
    $p(l|x)$, $p(\tilde{z}|x)$, and $q(y|\hat{z})$ can be expressed with $\theta$-, $\xi$-, and $\phi$- NNs, respectively. Rewriting the above via expectation operations yields the first part.
	The second part holds obviously.
\end{proof}
\bibliographystyle{IEEEtran}
\bibliography{main}

\begin{thebibliography}{10}
\providecommand{\url}[1]{#1}
\csname url@samestyle\endcsname
\providecommand{\newblock}{\relax}
\providecommand{\bibinfo}[2]{#2}
\providecommand{\BIBentrySTDinterwordspacing}{\spaceskip=0pt\relax}
\providecommand{\BIBentryALTinterwordstretchfactor}{4}
\providecommand{\BIBentryALTinterwordspacing}{\spaceskip=\fontdimen2\font plus
\BIBentryALTinterwordstretchfactor\fontdimen3\font minus
  \fontdimen4\font\relax}
\providecommand{\BIBforeignlanguage}[2]{{%
\expandafter\ifx\csname l@#1\endcsname\relax
\typeout{** WARNING: IEEEtran.bst: No hyphenation pattern has been}%
\typeout{** loaded for the language `#1'. Using the pattern for}%
\typeout{** the default language instead.}%
\else
\language=\csname l@#1\endcsname
\fi
#2}}
\providecommand{\BIBdecl}{\relax}
\BIBdecl

\bibitem{shi2021semantic}
G.~Shi, Y.~Xiao, Y.~Li, and X.~Xie, ``From semantic communication to
  semantic-aware networking: Model, architecture, and open problems,''
  \emph{IEEE Communications Magazine}, vol.~59, no.~8, pp. 44--50, 2021.

\bibitem{xiao2022imitation}
Y.~Xiao, Z.~Sun, G.~Shi, and D.~Niyato, ``Imitation learning-based implicit
  semantic-aware communication networks: Multi-layer representation and
  collaborative reasoning,'' \emph{IEEE Journal on Selected Areas in
  Communications}, vol.~41, no.~3, pp. 639--658, 2022.

\bibitem{10747747}
D.~Gündüz, M.~A. Wigger, T.-Y. Tung, P.~Zhang, and Y.~Xiao, ``Joint
  source–channel coding: Fundamentals and recent progress in practical
  designs,'' \emph{Proceedings of the IEEE}, pp. 1--32, 2024.

\bibitem{8723589}
E.~Bourtsoulatze \emph{et~al.}, ``Deep joint source-channel coding for wireless
  image transmission,'' \emph{IEEE Transactions on Cognitive Communications and
  Networking}, vol.~5, no.~3, pp. 567--579, 2019.

\bibitem{xu2023deep}
J.~Xu, T.-Y. Tung, B.~Ai, W.~Chen, Y.~Sun, and D.~D. G{\"u}nd{\"u}z, ``Deep
  joint source-channel coding for semantic communications,'' \emph{IEEE
  Communications Magazine}, vol.~61, no.~11, pp. 42--48, 2023.

\bibitem{shao2021learning}
J.~Shao, Y.~Mao, and J.~Zhang, ``Learning task-oriented communication for edge
  inference: An information bottleneck approach,'' \emph{IEEE Journal on
  Selected Areas in Communications}, vol.~40, no.~1, pp. 197--211, 2021.

\bibitem{xie2023robust}
S.~Xie, S.~Ma \emph{et~al.}, ``Robust information bottleneck for task-oriented
  communication with digital modulation,'' \emph{IEEE Journal on Selected Areas
  in Communications}, vol.~41, no.~8, pp. 2577--2591, 2023.

\bibitem{rulongICC}
R.~Wang, Y.~Zhou, Y.~Xiao, and Y.~Li, ``Multi-view semantic-aware
  communication: an information bottleneck perspective,'' in \emph{ICC 2025 -
  IEEE International Conference on Communications}, 2025.

\bibitem{10845799}
J.~Huang, K.~Yuan, C.~Huang, and K.~Huang, ``D2-jscc: Digital deep joint
  source-channel coding for semantic communications,'' \emph{IEEE Journal on
  Selected Areas in Communications}, pp. 1--1, 2025.

\bibitem{liu2022indirect}
J.~Liu, S.~Shao, W.~Zhang, and H.~V. Poor, ``An indirect rate-distortion
  characterization for semantic sources: General model and the case of gaussian
  observation,'' \emph{IEEE Transactions on Communications}, vol.~70, no.~9,
  pp. 5946--5959, 2022.

\bibitem{6651793}
T.~A. Courtade and T.~Weissman, ``Multiterminal source coding under logarithmic
  loss,'' \emph{IEEE Transactions on Information Theory}, vol.~60, no.~1, pp.
  740--761, 2014.

\bibitem{cover1999elements}
T.~M. Cover, \emph{Elements of information theory}.\hskip 1em plus 0.5em minus
  0.4em\relax John Wiley \& Sons, 1999.

\bibitem{9046817}
Y.~Kochman, O.~Ordentlich, and Y.~Polyanskiy, ``A lower bound on the expected
  distortion of joint source-channel coding,'' \emph{IEEE Transactions on
  Information Theory}, vol.~66, no.~8, pp. 4722--4741, 2020.

\bibitem{ccinlar2011probability}
E.~{\c{C}}inlar, \emph{Probability and stochastics}.\hskip 1em plus 0.5em minus
  0.4em\relax Springer, 2011.

\bibitem{lin2021learning}
T.~Lin, J.~Huh, C.~Stauffer, S.~N. Lim, and P.~Isola, ``Learning to ground
  multi-agent communication with autoencoders,'' \emph{Advances in Neural
  Information Processing Systems}, vol.~34, pp. 15\,230--15\,242, 2021.

\bibitem{yarotsky2017error}
D.~Yarotsky, ``Error bounds for approximations with deep relu networks,''
  \emph{Neural networks}, vol.~94, pp. 103--114, 2017.

\bibitem{sutton1999policy}
R.~S. Sutton, D.~McAllester, S.~Singh, and Y.~Mansour, ``Policy gradient
  methods for reinforcement learning with function approximation,''
  \emph{Advances in neural information processing systems}, vol.~12, 1999.

\end{thebibliography}
\end{document}